\documentclass[10pt,
a4paper]{article}

\usepackage{CJK}
\usepackage{CJKspace}
\usepackage{amsfonts}
\usepackage{mathrsfs}
\usepackage{comment}
\usepackage{amsmath}
\usepackage{amssymb}
\usepackage{amsthm}
\usepackage{amscd}
\usepackage{graphicx}
\usepackage{indentfirst}
\usepackage[all]{xy}
\usepackage{titlesec}
\usepackage[top=25.4mm, bottom=25.4mm, left=31.7mm, right=32.2mm]{geometry}
\usepackage{enumerate}
\usepackage{bm}
\newtheorem{theorem}{Theorem}[section]
\newtheorem{lemma}[theorem]{Lemma}

\newtheorem{corollary}[theorem]{Corollary}
\newtheorem{conjecture}[theorem]{Conjecture}
\newtheorem{definition}[theorem]{Definition}

\newtheorem{remark}[theorem]{Remark}

\newtheorem{example}[theorem]{Example}
\newtheorem{question}[theorem]{Question}

\newcommand{\ma}{\mathcal}

\newcommand{\s}{\subseteq}

\newcommand{\fr}{\frac}
\newcommand{\lc}{\lceil}
\newcommand{\rc}{\rceil}
\newcommand{\lf}{\lfloor}
\newcommand{\rf}{\rfloor}

\begin{document}
\begin{CJK*}{GBK}{song}

\title{Separating hash families: A Johnson-type bound and new constructions}
\author{
Chong Shangguan$^{\text{a}}$, Gennian Ge$^{\text{b,c,}}$\thanks{Corresponding author. Email address: gnge@zju.edu.cn. Research supported by the National Natural Science Foundation of China under Grant Nos. 11431003 and 61571310.
}
\\
\footnotesize $^{\text{a}}$ School of Mathematical Sciences, Zhejiang University, Hangzhou 310027, Zhejiang, China\\
\footnotesize $^{\text{b}}$ School of Mathematical Sciences, Capital Normal University, Beijing 100048, China\\
\footnotesize $^{\text{c}}$ Beijing Center for Mathematics and Information
Interdisciplinary Sciences,
Beijing 100048, China\\
}
\date{}
\maketitle

\begin{abstract}
    Separating hash families are useful combinatorial structures which are generalizations of many well-studied objects in combinatorics, cryptography and coding theory.
    In this paper, using tools from graph theory and additive number theory, we solve several open problems and conjectures concerning bounds and constructions for separating hash families.

    Firstly, we discover that the cardinality of a separating hash family satisfies a Johnson-type inequality. As a result, we obtain a new upper bound, which is superior to all previous ones.

    Secondly, we present a construction for an infinite class of perfect hash families. It is based on the Hamming graphs in coding theory and generalizes many constructions that appeared before. It provides an affirmative answer to both Bazrafshan-Trung's open problem on separating hash families and Alon-Stav's conjecture on parent-identifying codes.

    Thirdly, let $p_t(N,q)$ denote the maximal cardinality of a $t$-perfect hash family of length $N$ over an alphabet of size $q$. Walker and Colbourn conjectured that $p_3(3,q)=o(q^2)$. We verify this conjecture by proving $q^{2-o(1)}<p_3(3,q)=o(q^2)$. Our proof can be viewed as an application of Ruzsa-Szemer{\'e}di's (6,3)-theorem. We also prove $q^{2-o(1)}<p_4(4,q)=o(q^2)$. Two new notions in graph theory and additive number theory, namely rainbow cycles and $R$-sum-free sets, are introduced to prove this result. These two bounds support a question of Blackburn, Etzion, Stinson and Zaverucha.

    Finally, we establish a bridge between perfect hash families and hypergraph Tur{\'a}n problems. This connection has not been noticed before. As a consequence, many new results and problems arise.
\end{abstract}



\section{Introduction}

 Separating hash families are useful combinatorial structures introduced by Stinson, Wei and Chen \cite{ST08}. They are generalizations of many combinatorial objects, for example, perfect hash families, frameproof codes and codes with the identifiable parent property.

     Let us begin with some definitions.

   \begin{definition}

     Let $X$ and $Y$ be sets of cardinalities $n$ and $q$, respectively. We call a set $\ma{F}$ of $N$ functions $f:X\rightarrow Y$ an $(N;n,q)$-hash family.

   \end{definition}

   \begin{definition}

     Let $f:X\rightarrow Y$ be a function, and let pairwise disjoint subsets $C_1,C_2,\ldots,C_t\s X$. We say that $f$ separates $C_1,C_2,\ldots,C_t$ if $f(C_1),\ldots,f(C_t)$ are pairwise disjoint. In particular, we say that $f$ separates a subset $C\s X$ if $f(C)\s Y$ has $|C|$ distinct values.

   \end{definition}

   \begin{definition}

     Let $X$ and $Y$ be sets of cardinalities $n$ and $q$, respectively, and let $\ma{F}$ be an $(N;n,q)$-hash family of functions from $X$ to $Y$. We say that $\ma{F}$ is an $(N;n,q,\{w_1,\ldots,w_t\})$-separating hash family (which we will also denote as an $SHF(N;n,q,\{w_1,\ldots,w_t\})$) if it satisfies the following property: for all pairwise disjoint subsets $C_1,C_2,\ldots,C_t\s X$ with $|C_i|=w_i$ for $1\le i\le t$, there exists at least one function $f\in\ma{F}$ that separates $C_1,C_2,\ldots,C_t$. We call the multiset $\{w_1,\ldots,w_t\}$ the type of this separating hash family.

   \end{definition}

    For a positive integer $q$, we denote $[q]$ for the set $\{1,\ldots,q\}$. Without loss of generality, we may fix the alphabet set $Y$ to be the set of first $q$ positive integers. And for the sake of simplicity, we set $u=\sum_{i=1}^{t} w_i$ throughout this paper. To avoid trivial cases, we assume that $n>q$, $q\ge t\ge 2$ and $u\le n$.

    The concept of separating hash families was first introduced in the special case $t=2$ by Stinson, Trung and Wei \cite{secure} and then generalized by Stinson, Wei and Chen \cite{ST08}. This notion has relations with many well-studied objects in combinatorics, cryptography and coding theory, see \cite{blackburn08,ST08} for a detailed introduction. We will summarise some objects in which we are interested.

     \begin{itemize}
      \item If $w_1=w_2=\cdots w_t=1$, an $SHF(N;n,q,\{1,\ldots,1\})$ is known as a $t$-{\it perfect hash family}, which will be denoted as $PHF(N;n,q,t)$.
            Perfect hash families are basic combinatorial structures and have important applications in cryptography \cite{cyr1,cry2,ST01,secure}, database management \cite{data1}, circuit design \cite{circuit} and the design of deterministic analogues of probabilistic algorithms \cite{alog}.

      \item If $t=2$ with $w_1=1$ and $w_2=w$, an $SHF(N;n,q,\{1,w\})$ is known as a $w$-{\it frameproof code}. 
            The frameproof code is a kind of fingerprinting codes and has applications in the protection of copyrighted materials. See \cite{BL03,fpc1,ST01,fpc2} for results on frameproof codes.

      \item Codes with the {\it identifiable parent property} (or 2-$IPP$ codes) are separating hash families which are simultaneously of type $\{1,1,1\}$ and
            $\{2,2\}$, see \cite{ippa,ippe,ippc,ippd,ipp}.

     \end{itemize}

      Bounds and constructions for separating hash families are central problems in this research area. Given positive integers $N$, $q$ and $w_1,\ldots,w_t$,
      it is of interest how large the cardinality $n$ of the preimage set $X$ can be. We use $C(N,q,\{w_1,\ldots,w_t\})$ to denote this maximal cardinality.

      By a method known as grouping coordinates, the problem of bounding $C(N,q,\{w_1,\ldots,w_t\})$ can be reduced to bounding $C(u-1,q,\{w_1,\ldots,w_t\})$, since it has been observed in \cite{Trung2011,blackburn08,ST08} that $C(N,q,\{w_1,\ldots,w_t\})\le C(u-1,q^{\lceil N/(u-1)\rceil},\{w_1,\ldots,w_t\})$.

     In the literature, researchers are seeking for the minimal positive real number $\gamma$ such that $C(u-1,q,\{w_1,\ldots,w_t\})\le\gamma q$ holds for arbitrary $q$. The reader is referred to \cite{Trung2011,blackburn08,ST01,secure,ST08} for the attempts that have been made. In 2008, Stinson, Wei and Chen \cite{ST08} proved $C(3,q,\{1,1,2\})\le3q+2-2\sqrt{3m+1}$ and $C(3,q,\{2,2\})\le 4q-3$ for two special cases. In the same year, Blackburn, Etizon, Stinson and Zaverucha \cite{blackburn08} proved $C(u-1,q,\{w_1,\ldots,w_t\})\le(w_1w_2+u-w_1-w_2) q$, where $w_1,w_2\le w_i$ for $3\le i\le t$. In 2011, Bazrafshan and Trung \cite{Trung2011} proved the following theorem:

     \begin{theorem} {\rm(\cite{Trung2011})}\label{shf0}
        $C(u-1,q,\{w_1,\ldots,w_t\})\le(u-1)q.$
     \end{theorem}

     \noindent Moreover, they conjectured that (see Question \ref{trung}) $\gamma=u-1$ is the minimal real number such that the above bound holds for arbitrary $q$.

     We improve Theorem \ref{shf0} in various aspects, including some tighter bounds and asymptotically optimal constructions. The novelty of our work is that we develop two new approaches to study bounds and constructions for codes and hash families with the separating property. We will explain them in detail in the conclusion section of this paper.

     We state our main results as follows.

    \subsection{Separating hash families}

    Following the steps of previous papers \cite{Trung2011,blackburn08,ST08}, we discover an important property for separating hash families that the growth of $C(N,q,\{w_1,\ldots,w_t\})$ satisfies a Johnson-type inequality. Roughly speaking, $C(N,q,\{w_1,...,w_t\})\le q^l+\max\{u-1,C(N-l,q,\{w_1-1,...,w_t\})\}$ holds for every positive integer $l$ (see Lemma \ref{technical} below). As a result, we obtain the following new upper bound for separating hash families which is the best known one.

    \begin{theorem} \label{recusivebd}
        Suppose there exists an $SHF(N;n,q,\{w_1,\ldots,w_t\})$. Let $u=\sum_{i=1}^t w_i$ and let $1\le r\le u-1$ be the positive integer such that $N\equiv r \pmod{u-1}$.
        If $C(\lfloor N/(u-1)\rfloor,q,\{w_1,\ldots,w_t\})\ge u$, then it holds that $n\le rq^{\lceil N/(u-1)\rceil}+(u-1-r)q^{\lfloor N/(u-1)\rfloor}$.
    \end{theorem}

      A novelty of our proof is that we avoid the use of the grouping coordinates method, which has appeared in all previous proofs. The constraint $C(\lfloor N/(u-1)\rfloor,q,\{w_1,\ldots,w_t\})\ge u$ can be omitted when $N\ge u-1$ and $q\ge u$.

     For the coefficient $\gamma$ defined in Theorem \ref{shf0}, the authors of \cite{Trung2011} posed the following question:

    \begin{question} \label{trung} {\rm(\cite{Trung2011})}
        Is there any type $\{w_1,\ldots,w_t\}$ for which the constant $(u-1)$ in Theorem \ref{shf0} can be replaced by another constant strictly smaller than $(u-1)$?
    \end{question}

    We give a negative answer to their question by presenting the following construction:

    \begin{theorem}\label{const}
        There exists a $PHF(N;Nq^{N-1},q^{N-1}+(N-1)q^{N-2},N+1)$ for any integer $q\ge 2$ and $N\ge 2$. As a consequence, $\gamma=u-1$ is the minimal real number such that
        $C(u-1,q,\{w_1,\ldots,w_t\})\le\gamma q$ holds for arbitrary $q$.
    \end{theorem}

    To see that our construction is actually a negative answer to Question \ref{trung}, one just needs to notice that a $u$-perfect hash family is also $\{w_1,\ldots,w_t\}$-separating for arbitrary $\sum_{i=1}^t w_i=u$. If we set $N=u-1$ then our construction implies the existence of an $SHF(u-1;n,q,\{w_1,\ldots,w_t\})$ such that $\lim_{q\rightarrow\infty}\fr{n}{q}=u-1$ holds for arbitrary $\sum_{i=1}^t w_i=u$. Therefore, the constant $\gamma$ can never be less than $u-1$.

    \subsection{Codes with the identifiable parent property}

    We have mentioned 2-IPP codes before and the notion was generalized to codes with the $t$-identifiable parent property ($t$-IPP codes) in \cite{ST01}. We postpone the definition to Section 2 for the sake of saving space.

    Let $i_t(N,q)$ denote the maximal cardinality of a $t$-IPP code of length $N$ over an alphabet of size $q$. Let $v=\lfloor(t/2+1)^2\rfloor$. One can verify that $i_t(N,q)\le i_t(v-1,q^{\lc N/(v-1)\rc})$ (just as the case for separating hash families). Thus the problem of bounding $i_t(N,q)$ can be reduced to bounding $i_t(v-1,q)$. Alon and Stav \cite{ippc} proved that $i_t(v-1,q)\le (v-1)q$, and they conjectured:

    \begin{conjecture}\label{alonstav}{\rm(\cite{ippc})}
        There are constructions showing that $(v-1)$ is the best constant in the inequality $i_t(v-1,q)\le(v-1)q.$
    \end{conjecture}

    Our Theorem \ref{const} not only answers Question \ref{trung} but also verifies this conjecture, since it was observed in \cite{ippc,ST01} that a $v$-perfect hash family also satisfies the $t$-identifiable parent property.

    \subsection{Perfect hash families}

    As claimed in \cite{blackburn08}, the exponent $\lc N/(u-1)\rc$ in the bound of Theorem \ref{recusivebd} is realistic. We can understand this in two aspects. On the one hand, a probabilistic construction of Blackburn \cite{perfect2} showed that for any fixed $u$ and any positive real number $\delta$ such that $\delta<N/(u-1)$, there exists a $PHF(N;\lfloor q^{\delta}\rfloor,q,u)$ whenever $q$ is sufficiently large. On the other hand, let $p_t(N,q)$ denote the maximal cardinality of a $PHF(N;n,q,t)$, it was respectively observed in \cite{ippc,perfecthasing,nilli} that $p_u(N,q)\ge (c_uq)^{N/(u-1)}$ holds for some constant $c_u$. So we can conclude that the exponent $\lc N/(u-1)\rc$ is tight when $(u-1)|N$.

    But the problem becomes much more difficult when $(u-1)\nmid N$. It is not known that whether the exponent is tight. Even for the smallest case, $u=3$ and $N=3$, Walker and Colbourn \cite{N=3t=3} posed the following conjecture:

    \begin{conjecture} \label{walker}{\rm(\cite{N=3t=3})}
        $p_3(3,q)=o(q^2).$
    \end{conjecture}
    \noindent Note that Theorem \ref{recusivebd} shows $p_3(3,q)=O(q^2)$. A recent paper \cite{fuji2015perfect} showed that $p_3(3,q)=\Omega(q^{5/3})$. Results from finite geometry were used to construct such families. There is still a huge gap between the upper and lower bounds.
    For general types of separating hash families, Blackburn et al. \cite{blackburn08} asked a similar question:

    \begin{question}\label{nmid}{\rm(\cite{blackburn08})}
        Let $N$ and $w_i$ be fixed integers. If $(u-1)\nmid N$, then for sufficiently large $q$ and arbitrary small $\epsilon>0$, does there exist an $SHF(N;n,q,\{w_1,\ldots,w_t\})$ such that $n\ge q^{\lceil N/(u-1)\rceil-\epsilon}?$
    \end{question}

     We prove Conjecture \ref{walker} in Section 5 (see Theorems \ref{p_3(3,q)=f_3(3q,6,3)+O(q)} and \ref{optimal3} below). We find that perfect hash families are closely related to a hypergraph Tur{\'a}n problem. With some transformations, Walker-Colbourn's conjecture can be proved by a direct application of the famous (6,3)-theorem of Ruzsa and Szemer{\'e}di \cite{RS}. In fact, we show $$q^{2-\epsilon}<p_3(3,q)=o(q^2)$$ holds for sufficiently large $q$ and arbitrary $\epsilon>0$. We also prove $$q^{2-\epsilon}<p_4(4,q)=o(q^2)$$ (see Theorem \ref{optimal4} below). Two new notions in graph theory and additive number theory, namely rainbow cycles and $R$-sum-free sets, are introduced to prove this result. One can see that these two bounds suggest that there may be a positive answer to Question \ref{nmid}.

     \subsection{Organization}

    The rest of this paper is organised as follows. Section 2 is for some preparations. Theorem \ref{recusivebd} is proved in Section 3 and Theorem \ref{const} is proved in Section 4. The subsequent sections will focus on perfect hash families. We prove $q^{2-o(1)}<p_3(3,q)=o(q^2)$ in Section 5. As an application of the Johnson-type bound, this result will be extended to $p_t(t,q)$ and related separating hash families. We prove $q^{2-o(1)}<p_4(4,q)=o(q^2)$ in Section 6. In Section 7 we will build the connection between perfect hash families and a class of hypergraph Tur{\'a}n problems. Section 8 consists of some concluding remarks and open problems.

    \section{Preliminaries}

    In this section, we will introduce some notations and terminology. We will also introduce some simple lemmas that will be used in the subsequent sections.

    \subsection{Separating hash families and IPP codes}

    The matrix representation of a separating hash family is very useful when discussing its properties. An $(N;n,q)$-hash family can be described as an $N\times n$ matrix on $q$ symbols, which will be usually denoted as $M$. The rows of $M$ correspond to the functions in the hash family and the columns of $M$ correspond to the elements of $X$.
    The entry of $M$ in row $f\in\ma{F}$ and column $x\in X$ is just $f(x)\in Y.$ We denote the entry of $M$ as $M(f,x)$ for $f\in\ma{F}$, $x\in X$ or $M(i,j)$ for $1\le i\le N$, $1\le j\le n$.

    The matrix representation of an $SHF(N;n,q,\{w_1,\ldots,w_t\})$ satisfies the following property: given disjoint sets of columns $C_1,\ldots,C_t$, where $|C_i|=w_i$ for $1\le i\le t$, there exists a row $r$ of $M$ such that
    $$\{M(r,x):~x\in C_i\}\cap\{M(r,x):~x\in C_j\}=\emptyset$$
    \noindent for all $i\neq j$. We say row $r$ separates a subset of columns $C\s X$ if $\{M(r,x):x\in C\}$ has exactly $|C|$ distinct values in $Y$. The column $x$ of $M$ will be written as a $q$-ary vector of length $N$, $x=(x(1),x(2),\ldots,x(N))$, where $x(i)\in [q]$ for $i\in[N]$. For a subset $L$ of the rows of $M$, the coordinates of $x$ restricted to $L$ give a word of length $|L|$, which is denoted as $x|_{L}=(x(i_1),x(i_2),\ldots,x(i_{|L|}))$, where $i_j,~1\le j\le |L|$ are the row indices. We say a column $x\in X$ of $M$ has a unique coordinate $i$ if for any other column $y\in X$, $y\neq x$, it holds that $y(i)\neq x(i)$. If there is no confusion, we will not distinguish between a hash family and its representation matrix.

    Next we will introduce the definition of IPP codes.

    Let $\ma{C}\s Y^N$ be a code of length $N$ and let $D\s\ma{C}$ be a set of codewords. The set of descendants of $D$, denoted as $desc(D)$, is defined by

    $$desc(D)=\{d\in Y^N:for~all~i\in\{1,2,\ldots,N\},~d(i)=x(i)~for~some~x\in D\}.$$

    \noindent A set $D\s\ma{C}$ is said to be a parent set of a word $d\in Y^N$ if $d\in desc(D)$. For $d\in Y^N$, let $\ma{P}_t(d)$ denote the collection of parent sets of $d$ such that $|D|\le t$ and $D\s\ma{C}$. Then we call $\ma{C}\s Y^N$ a $t$-IPP code if for all $d\in Y^N$, either $\ma{P}_t(d)=\emptyset$ or

    $$\cap_{D\in\ma{P}_t(d)} D\neq\emptyset.$$

    \subsection{Graph theory}

    We will use the notion of Hamming graphs when constructing perfect hash families in Section 4. Let $k$ and $q$ be positive integers, the Hamming graph (see \cite{algcomb} for details) $H(k,q)$ has the set of all $k$-tuples from an alphabet of $q$ symbols as its vertex set, and two $k$-tuples are adjacent if and only if they differ in exactly one coordinate position. This graph is also known as the $q$-ary hypercube of dimension $k$. Here we will fix this $q$-symbol alphabet set to be $[q]$.

    When speaking about a hypergraph we mean a pair $\ma{G}=(V(\ma{G}),E(\ma{G}))$, where the vertex set $V(\ma{G})$ is identified as the set of first integers $[n]$ and the edge set $E(\ma{G})$ is identified as a collection of subsets of $[n]$. $\ma{G}$ is said to be linear if for all distinct $A,B\in E(\ma{G})$ it holds that $|A \cap B|\le1$. We say $\ma{G}$ is $r$-uniform if $|A|=r$ for all $A\in E(\ma{G})$.

    An $r$-uniform hypergraph $\ma{G}$ is $r$-partite if its vertex set $V(\ma{G})$ can be colored in $r$ colors in such a way that no edge of $\ma{G}$ contains two vertices of the same color. In such a coloring, the color classes of $V(\ma{G})$, the sets of all vertices of the same color, are called parts of $\ma{G}$. In this paper we mainly concern $r$-uniform $r$-partite hypergraphs with equal part size $q$. We will see later that the edge set of such hypergraph is equivalent to an $r\times|E(\ma{G})|$ matrix over a $q$-symbol alphabet.

    Given a set $\ma{H}$ of $r$-uniform hypergraphs, an $\ma{H}$-free $r$-uniform hypergraph is a graph containing none of the members of $\ma{H}$. The Tur{\'a}n number $ex_r(n,\ma{H})$ denotes the maximum number of edges in an $\ma{H}$-free $r$-uniform hypergraph on $n$ vertices. In this paper, we will talk about several hypergraph Tur{\'a}n problems.

    Brown, Erd{\H{o}}s and S{\'o}s \cite{bes1,bes2} introduced the function $f_r(n,v,e)$ to denote the maximum number of edges in an $r$-uniform hypergraph on $n$ vertices which does not contain $e$ edges spanned by $v$ vertices. In other words, in such hypergraphs the size of the union of arbitrary $e$ edges is at least $v+1$. These hypergraphs are called $G(v,e)$-free (more precisely, $G_r(v,e)$-free). The famous (6,3)-theorem of Ruzsa and Szemer{\'e}di \cite{RS} pointed out that

    \begin{equation}\label{1}
        \begin{aligned}
            n^{2-o(1)}<f_3(n,6,3)=o(n^2).
        \end{aligned}
    \end{equation}

    \noindent This was extended by Alon and Shapira \cite{t=3} to

    \begin{equation}\label{2}
        \begin{aligned}
            n^{k-o(1)}<f_r(n,3(r-k)+k+1,3)=o(n^k).
        \end{aligned}
    \end{equation}


    \noindent These bounds will be used when considering problems about perfect hash families in the sequel. For more results on Tur{\'a}n problems of this type, see \cite{Furediconst} and the references therein.

    The definitions of $f_r(n,v,e)$ can be restricted to the case for $r$-uniform $r$-partite hypergraphs with equal part size $q$. We use $f_r^*(q,v,e)$ to denote the corresponding formula. Note that $f_r^*(q,v,e)\le f_r(rq,v,e)$.

    In the literature, there are several definitions of hypergraph cycles. The one we use in this paper was introduced by Berge \cite{hypercycle2,hypercycle1}. For $k\ge2$, a cycle in a hypergraph $\ma{G}$ is an alternating sequence of vertices and edges of the form $v_1,E_1,v_2,E_2,\ldots,v_k,E_k,v_1$ such that

    \begin{itemize}
      \item [(a)]$v_1,v_2,\ldots,v_k$ are distinct vertices of $\ma{G}$,
      \item [(b)]$E_1,E_2,\ldots,E_k$ are distinct edges of $\ma{G}$,
      \item [(c)]$v_i,v_{i+1}\in E_i$ for $1\le i\le k-1$ and $v_k,v_1\in E_k$.
    \end{itemize}


    Next we will introduce the definition of rainbow cycles. Note that in the literature ``rainbow cycles" always stand for edge-colorings, but in this paper we consider rainbow cycles due to vertex colorings. Given a hypergraph $\ma{G}$ and a vertex-coloring of $\ma{G}$, a subgraph $\ma{H}\s\ma{G}$ is called a rainbow subgraph of $\ma{G}$ if all joint vertices in $\ma{H}$ have different colors. In other words, for arbitrary distinct vertices $x,y\in\{A\cap B:A,B\in E(\ma{H})\}$, $x$ and $y$ are colored by different colors. This definition is most meaningful when discussing linear hypergraphs. Let $\ma{G}$ be an $r$-uniform $r$-partite linear hypergraph, a $k$-cycle $v_1,E_1,v_2,E_2,\ldots,v_k,E_k,v_1$ is said to be a rainbow cycle of $\ma{G}$ if $v_1,\ldots,v_k$ locate in different parts of $V(\ma{G})$. For $r$-partite graphs, a rainbow $k$-cycle exists only if $k\le r$.

    Let $\ma{G}$ be an $r$-uniform $r$-partite linear hypergraph with equal part size $q$. Assume that $\ma{G}$ can not have rainbow cycles, then we use $g_r^*(q)$ to denote the maximal number of edges that can be contained in $\ma{G}$. Lemma \ref{rainbow1} shows that hypergraphs with large $g_r^*(q)$ can be used to construct good perfect hash families.

    \subsection{Additive number theory}
        It has been shown in \cite{ippe,Furediconst} that tools from additive number theory can be used to construct codes with some specified properties. We will introduce some notions from additive number theory.

        Assume $m_1,m_2,m_3\in M\s[q]$ and $c_1,c_2$ are positive integers such that $c_1+c_2\le r$, we call the set $M$ $r$-sum-free if the equation $$c_1m_1+c_2m_2=(c_1+c_2)m_3$$ has no solution except the one with $m_1=m_2=m_3$. A result proved by Erd{\H{o}}s, Frankl and R$\ddot{o}$dl \cite{erdos2} and Ruzsa \cite{ruzsa} will be needed.

        \begin{lemma}{\rm(\cite{erdos2,ruzsa})}\label{additive}
        For arbitrary positive integer $r$ there exists a $\gamma_r>0$ such that for any integer $q$, one can find an $r$-sum-free subset $M\s[q]$ with $|M|>qe^{-\gamma_r\sqrt{\log q}}$.
        \end{lemma}

        \noindent Note that the case $r_1=r_2=1$ was originally proved by Behrend \cite{Behrend}.

        A linear equation with integer coefficients $$\sum_{i=1}^k a_ix_i=0$$ in the unknowns $x_i$ is homogeneous if $\sum_{i=1}^k a_i=0$. We say that $M\s[q]$ has no nontrivial solution to above equation, if whenever $m_i\in M$ and $\sum_{i=1}^k a_im_i=0$, it follows that all $m_i$'s are equal. Note that if $M$ has no nontrivial solution to above function, then the same holds for any shift $(M+x)\cap[q]$ with $x\in\mathbb{Z}$, where $M+x:=\{m+x:m\in M\}$. This property suggests that one can use probabilistic method to construct sets with no nontrivial solution to a system of homogeneous linear equations. Note that this definition of the nontrivial solution is a simplification of the original one of Ruzsa \cite{ruzsa}.

        Now we will generalize the definition of the $r$-sum-free set. Given a set $R=\{b_1,\ldots,b_r\}$ of $r$ distinct nonnegative integers. A set $M$ is said to be $R$-sum-free if for any $3\le k\le r$ and any $k$-element subset $S=\{b_{j_1},b_{j_2},\ldots,b_{j_k}\}\s R$, the equation
        $$(b_{j_2}-b_{j_1})m_1+(b_{j_3}-b_{j_2})m_2+\cdots+(b_{j_k}-b_{j_{k-1}})m_{k-1}+(b_{j_1}-b_{j_k})m_k=0$$ has no solution in $M$ except the trivial one $m_1=m_2=\cdots=m_k$. The rank of $R$ is defined to be the maximal difference between the elements of $R$:
        $$r(R)=\max_{1\le i<j\le r} |b_i-b_j|.$$

        We are interested in $R$-sum-free sets $M\s[q]$ with relatively small rank, namely, $r(R)=o(q^{\epsilon})$ for arbitrary $\epsilon>0$. Lemma \ref{rainbow2} shows that  $R$-sum-free sets can be used to construct hypergraphs with large $g_r^*(q)$.

        \subsection{Some lemmas}


    The following lemma is a variant of a result of Erd{\"o}s and Kleitman \cite{erdos-kleitman}.

    \begin{lemma}\label{erdos-kleitman}
        Every $r$-uniform hypergraph $\ma{G}$ contains an $r$-uniform $r$-partite hypergraph $\ma{H}$ with equal part size $q$ or $q+1$ such that
        $$\fr{|E(\ma{H})|}{|E(\ma{G})|}\ge\fr{r!}{r^r}.$$
    \end{lemma}

    \begin{proof}
        Let $|V(\ma{G})|=n$, take $q$ to be the integer such that $rq\le n<r(q+1)$. We only prove the lemma for $n=rq$, otherwise we can set the part size of the desired subgraph to be $q+1$. It suffices to find a partition $\pi$ of $V(\ma{G})$ with $\pi=\{B_1,\ldots,B_r\}$ and $|B_i|=q$ for $1\le i\le r$, such that $\ma{F}_{\pi}=\{A\in E(\ma{G}):|A\cap B_i|=1~for~all~1\le i\le r\}$ contains the desired number of edges. Let $P(\ma{G})$ denote the collection of all appropriate partitions of $V(\ma{G})$. Let us count the number of the pairs
        $N:=|\{(A,\pi):A\in E(\ma{G}),~\pi\in P(\ma{G}),~|A\cap B_i|=1~for~every~B_i\in\pi\}|$. One can compute that any $A\in E(\ma{G})$ is contained in $\fr{|P(\ma{G})|\cdot q^r}{\binom{rq}{r}}$ members of $P(\ma{G})$ satisfying the desired property. Therefore, by double counting, there exists a $\pi\in P(\ma{G})$ such that $\ma{F}_{\pi}$ contains at least

        $$\fr{|E(\ma{G})|\cdot|P(\ma{G})|\cdot q^r/\binom{rq}{r}}{|P(\ma{G})|}=\fr{|E(\ma{G})|\cdot q^r}{\binom{rq}{r}}$$

        \noindent members of $E(\ma{G})$. Then this specified $\pi$ will induce an $r$-uniform $r$-partite hypergraph $\ma{H}$ containing the desired number of edges.
    \end{proof}

    This lemma implies that for any $r$-uniform hypergraph $\ma{G}$ with sufficiently large $|V(\ma{G})|$, there exists an $r$-partite subgraph $\ma{H}\s\ma{G}$ such that $|E(\ma{H})|$ and $|E(\ma{G})|$ are of the same order of magnitude. In other words, one can infer $f_r(rq,v,e)=\Theta(f_r^*(q,v,e))$ by Lemma \ref{erdos-kleitman}.

    Another simple lemma will be used.

    \begin{lemma}\label{treelemma}
        Suppose $G$ is a finite graph with $n$ vertices. If $G$ has no cycles, then $G$ can have at most $n-1$ edges.
     \end{lemma}

    \begin{proof}
       $G$ must have a vertex with degree one since every path in $G$ is finite and must have an end point. Choose a vertex in $G$ with degree one, then the statement follows trivially by applying induction on $|V|$.
    \end{proof}

    With some reformulations, one can combine Lemma 3.2 and Corollary 3.3 of Alon, Fischer and Szegedy \cite{ippe} to prove the following result:

        \begin{lemma}{\rm(\cite{ippe})}\label{t=4p=4}
            There exists a set $M\s\{0,1,\ldots,\lf (q-1)/(\mu+5)\rf\}$ satisfying $$ |M|\ge qe^{-\gamma(\log q)^{3/4}}$$ such that $M$ has no non-trivial solution to all the following equations

        \begin{equation}\label{3}
          \left\{
            \begin{aligned}
                &2m_1+3m_2+\mu m_3-(\mu+5)m_4&=0\\
                &5m_1+(\mu+3)m_2-3m_3-(\mu+5)m_4&=0\\
                &5m_1+\mu m_2-2m_3-(\mu+3)m_4&=0\\
                &2m_1+3m_2-5m_3 &=0\\
                &5m_1+\mu m_2-(\mu+5)m_3 &=0\\
                &2m_1+(\mu+3)m_2-(\mu+5)m_3 &=0\\
                &3m_1+\mu m_2-(\mu+3)m_3 &=0\\
            \end{aligned}
          \right.
        \end{equation}
        \noindent where $\gamma$ is a constant and $\mu=\lc 2^{\sqrt{\log q}}\rc$.
        \end{lemma}

        \begin{proof}[Sketch of the proof]
            Using the technique introduced in the proof of Lemma 3.2 of \cite{ippe}, for $1\le i\le7$, one can prove that there exists a set
            $M_i\s\{0,1,\ldots,\lf (q-1)/(\mu+5)\rf\}$ and a constant $\gamma_i$ satisfying $$ |M_i|\ge qe^{-\gamma_i(\log q)^{3/4}}$$ such that $M_i$ has no nontrivial solution to the $i$-th equation in the above system. In order to prove the existence of the set $M$ which has no nontrivial solution to all equations, we can apply a probabilistic method. Take six integers $x_i$ such that $-\lf (q-1)/(\mu+5)\rf\le x_i\le \lf (q-1)/(\mu+5)\rf$, $2\le i\le 7$, randomly, uniformly and independently. $M=M_1\cap(M_2+x_2)\cap\cdots\cap(M_7+x_7)$ has no nontrivial solution to any of the above equations. Since $M_i+x_i\in[-\lf (q-1)/(\mu+5)\rf,2\lf (q-1)/(\mu+5)\rf]$ for each $2\le i\le 7$, then one can compute that every $m\in M_1$ has probability at least $e^{-\sum_{i=2}^7 \gamma_i(\log q)^{3/4}}$ to lie in the intersection. Therefore, the result follows from the linearity of the expectation, where $|M|\ge qe^{-\gamma(\log q)^{3/4}}$ with $\gamma\le\sum_{i=1}^7 \gamma_i$.
        \end{proof}

    \section{A Johnson-type upper bound}

   The aim of this section is to establish a Johnson-type bound for separating hash families and we will use it to prove Theorem \ref{recusivebd}. To establish this bound, the idea is to delete some rows and corresponding carefully chosen columns from the representation matrix of the separating hash family. Our goal is to show the remaining submatrix satisfies some weaker separating property. We call this recursive bound a ``Johnson-type bound" due to its similarity with the traditional recursive Johnson bound in coding theory. Note that we always use $M$ to denote the representation matrix of a separating hash family.

    \begin{lemma} \label{technical}
        Let $1\le l\le N$ be a positive integer, then it holds that $C(N,q,\{w_1,...,w_t\})\le q^l+\max\{u-1,C(N-l,q,\{w_1-1,...,w_t\})\}$. In fact, in the right hand side of the inequality we can choose the minus of 1 to be after an arbitrary $w_i,~1\le i\le t$.
    \end{lemma}

  \begin{proof}
      Choose arbitrary $l$ rows of $M$ and let $L$ denote the collection of these chosen rows. Denote $\ma{A}\s Y^l$ the maximal collection of columns whose restrictions to $L$ are all distinct (we just choose one column if there are several columns with the same restrictions to $L$). It is easy to see $|\ma{A}|\le q^l$ since there are at most $q^l$ distinct words of length $l$. Delete these $l$ rows and the columns contained in $\ma{A}$ from $M$. Let $M'$ denote the remaining submatrix. Then $M'$ is a $q$-ary $(N-l)\times(n-|\ma{A}|)$ matrix. If $n-|\ma{A}|\le u-1$, we are done. Otherwise it suffices to show $M'$ is a representation matrix of a separating hash family of type $\{w_1,\ldots,w_i-1,\ldots,w_t\}$ for arbitrary $1\le i\le t$.

      Assume the contrary, $M'$ is not $\{w_1,\ldots,w_i-1,\ldots,w_t\}$-separating for some $1\le i\le t$. Without loss of generality, we set $i=1$. Then there exist $t$ subsets $C_1,\ldots,C_t$ of the columns of $M'$ with $|C_1|=w_1-1$ and $|C_i|=w_i$ for $2\le i\le t$, such that no row of $M'$ can separate $C_1,\ldots,C_t$. Let $c$ be an arbitrary column of $C_2$ and let $c'$ be a column in $\ma{A}$ such that $c'|_{L}=c|_{L}$. Such $c'\in\ma{A}$ must exist by our definition of $\ma{A}$. Consequently, no row can separate $C_1\cup\{c'\},C_2,\ldots,C_t$ in the original matrix $M$, which contradicts the fact that $M$ is $\{w_1,\ldots,w_t\}$-separating. Thus $M'$ satisfies the desired separating property and the lemma follows from $n-|\ma{A}|\le C(N-l,q,\{w_1-1,...,w_t\}$ and $|\ma{A}|\le q^l$.
 \end{proof}

    \begin{remark}
        This lemma is obviously an extension of Lemma 2 of \cite{Trung2011}. We think this Johnson-type bound is very interesting and important since it points out the information hidden in the structure of separating hash families.
    \end{remark}

    As the first application of Lemma \ref{technical}, we will use it to prove Theorem \ref{recusivebd}.
    Note that we can omit the constraint $C(\lfloor N/(u-1)\rfloor,q,\{w_1,\ldots,w_t\})\ge u$ in the theorem by introducing a maximum term in the expression of the upper bound (just as the case in Lemma \ref{technical}). And $C(\lfloor N/(u-1)\rfloor,q,\{w_1,\ldots,w_t\})\ge u$ always holds for $N\ge u-1$ and sufficiently large $q$, for example, $q\ge u$.

  \begin{proof}[\textbf{Proof of Theorem \ref{recusivebd}}]
      One can verify that $N=r\lceil N/(u-1)\rceil+(u-1-r)\lfloor N/(u-1)\rfloor$. We apply Lemma \ref{technical} repeatedly for $u-1$ times, in which $l$ is chosen to be $\lceil N/(u-1)\rceil$ ($r$ times) and $\lfloor N/(u-1)\rfloor$ ($u-1-r$ times), respectively. 
      The theorem follows from a simple fact that $C(0,q,\{1\})=0$.
 \end{proof}

     \begin{remark}
        It is not hard to see our bound is an improvement of Theorem \ref{shf0}, and hence an improvement of \cite{Trung2011,blackburn08}. One can see $\lceil N/(u-1)\rceil$ is the best exponential term that can be obtained by our method, since to reduce the exponential term, one should reduce the maximum value of $l$ involved in the deletions. In other words, we should find a finer partition of $[N]$ and hence more deletion rounds are needed. However, at most $(u-1)$ deletion rounds can be used, because $C(N,q,\{w_1,\ldots,w_t\})$ can be arbitrary large if $t=1$ and $N>0$.
     \end{remark}

     \begin{remark}
        Since the frameproof code is a special class of separating hash families, it is not surprising to see our bound contains Theorem 1 of \cite{BL03} as a special case. By Constructions 2 and 3 in \cite{BL03}, one can find that Theorem \ref{recusivebd} is asymptotically optimal when $q\ge N$, $\{w_1,\ldots,w_t\}=\{1,w\}$, $N\equiv 1 \pmod{w}$, or $q=\Omega(N^2)$, $\{w_1,\ldots,w_t\}=\{1,2\}$. The following section presents a construction which shows Theorem \ref{recusivebd} is also asymptotically optimal when $N=u-1$.
     \end{remark}

\section{A construction for $t$-perfect hash families with $t-1$ rows}

    The aim of this section is to present a construction for $PHF(N;Nq^{N-1},q^{N-1}+(N-1)q^{N-2},N+1)$ for arbitrary positive integer $q\ge2$ and $N\ge 2$.

    One nice feature of our construction is that it is a generalization of many previous ones. When $N=2$, the construction of $PHF(2;2q,q+1,3)$ has appeared in \cite{N=2first,N=3t=3}. And when $N=3$, the construction of $PHF(3;3q^2,q^2+2q,4)$ has appeared in numerous papers, for example, Hollmann et al. \cite{ipp}, Blackburn \cite{BL03perfect}, Stinson et al. \cite{ST08} and Bazrafshan et al. \cite{Trung2011}.

   Let us begin with $N=3$ as a simple example to illustrate our idea.

    \begin{example}{\rm(\cite{Trung2011,BL03perfect,ipp,ST08})}
        There exists a $PHF(3;3q^2,q^2+2q,4)$ for any integer $q\ge 2$.
    \end{example}

     \begin{proof}
        We first construct a $3\times q^2$ submatrix, in which the alphabet set is the $(q^2+2q)$-element set defined as $\{(x,y),(x,0),(0,y):1\le x,y\le q,~x,y\in\mathbb{Z} \}$,

        \begin{center}
        $\left(
            \begin{array}{cccccccccccccccccccccccccccccccc}
              (1,1) & (1,2) & \cdots & (1,q) & (2,1) & \cdots & (2,q) & \cdots & (q,1) & \cdots & (q,q) \\
              (0,1) & (0,2) & \cdots & (0,q) & (0,1) & \cdots & (0,q) & \cdots & (0,1) & \cdots & (0,q) \\
              (1,0) & (1,0) & \cdots & (1,0) & (2,0) & \cdots & (2,0) & \cdots & (q,0) & \cdots & (q,0) \\
            \end{array}
        \right)$.
        \end{center}

        \noindent We denote the three rows of this submatrix as $A_0,~A_1,~A_2$, respectively.
        Then the representation matrix of the desired perfect hash family can be presented as follows:

        \begin{center}
        $\left(
              \begin{array}{ccc}
                A_0 & A_2 & A_1 \\
                A_1 & A_0 & A_2 \\
                A_2 & A_1 & A_0 \\
              \end{array}
        \right).$
        \end{center}

        \noindent We can easily see it is a $3\times 3q^2$ matrix over an alphabet of size $q^2+2q$. One can verify (or see the proof of Theorem \ref{const} below) that it is indeed a representation matrix of a 4-perfect hash family.
    \end{proof}

     In the above matrix $A_0$ acts like an identity map that preserves each element in $\{(x,y):1\le x,y \le q,~x,y\in\mathbb{Z}\}$, while $A_i$, $i=1,~2$, acts like a projection that projects the $i$-th entry of $(x,y)$ to zero. Actually, the idea behind this simple construction can be generalized.


     Recall the definition of the Hamming graphs in Section 2. Take a $q$-ary hypercube $\ma{A}$ of dimension $k$, then $|V(\ma{A})|=q^k$. For $1\leq i\le k$ and arbitrary $\alpha=(\alpha(1),\ldots,\alpha(k))\in V(\ma{A})$, define $\pi_i$ to be the map that sets $\alpha(i)$ into zero but preserves all other coordinates of $\alpha$. We say $\pi_i$ separates a set $S\s V(\ma{A})$ if $\pi_i(\alpha)\neq\pi_i(\beta)$ for arbitrary distinct $\alpha,~\beta\in S$. Proposition 1 of \cite{BL03perfect} establishes an important property of these maps. We present the proof here for the sake of reader's convenience.

    \begin{lemma}{\rm(\cite{BL03perfect})}{\label{construction}}
        Let $S\s V(\ma{A})$ be an arbitrary $t$-element subset with $t\le k$, then $S$ is separated by at least $k-t+1$ of the functions $\pi_1,\ldots,\pi_k$.
    \end{lemma}

    \begin{proof}
        Assume the contrary. Without loss of generality, let $S=\{\alpha_1,\ldots,\alpha_t\}$ and let $\pi_1,\ldots,\pi_t$ be the $t$ functions which can not separate $S$. Define a colored graph $G=(V,E)$ by $V=S$ and connect $\alpha,\beta\in V$ by an edge of color $i$ if $\pi_i(\alpha)=\pi_i(\beta)$. Note that graph $G$ is a subgraph of a Hamming graph. Since $\pi_1,\ldots,\pi_t$ can not separate $S$, then for every $i\in[t]$, there exist $1\le j<l\le t$ such that $\pi_i(\alpha_j)=\pi_i(\alpha_l)$. So $G$ contains a subgraph $G^{'}=(V,E^{'})$ with $t$ vertices and $t$ edges of distinct colors. By Lemma \ref{treelemma} we can deduce that $G^{'}$ contains a cycle which can be denoted as $(\alpha_1,\alpha_2,\ldots,\alpha_c)$, where $c$ is an integer such that $1\leq c\leq t$.

        We are done if we can show such cycle must not exist. Assume that the edge between $\alpha_1$ and $\alpha_2$ is colored by the $i$-th color. Then $\alpha_1$ and $\alpha_2$ must differ in their $i$-th coordinate. Since every edge in this cycle is of distinct color and every pair of connected vertices differ in exactly one coordinate, then for every $j\in\{2,3,\ldots,c\}$, $\alpha_j$ and $\alpha_{j+1}$ must agree in their $i$-th coordinate. In particular, $\alpha_{c+1}$ is recognised as $\alpha_1$, which implies that $\alpha_2(i)=\alpha_3(i)=\ldots=\alpha_c(i)=\alpha_1(i)$. Thus the desired contradiction follows.
    \end{proof}

    The following lemma is an easy consequence of above lemma.

    \begin{lemma}\label{lastlemma}
        Let $\pi_i~(1\le i\le k)$ be the functions defined as above and let $\pi_0$ denote the identity map which satisfies $\pi_0(\alpha)=\alpha$ for every $\alpha\in V(\ma{A})$. Suppose $S\s V(\ma{A})$ is a $t$-element subset with $t\le k+1$, then at most $t-1$ of the functions $\pi_0,\pi_1,\ldots,\pi_k$ can not separate $S$.
    \end{lemma}

    \begin{proof}
        Apply Lemma \ref{construction} and remember the fact that $\pi_0$ separates every subset of $V(\ma{A})$.
    \end{proof}



    Now we can prove Theorem \ref{const}.

    \begin{proof}[\textbf{Proof of Theorem \ref{const}}]

        Take a $q$-ary hypercube $\ma{A}$ of dimension $N-1$. Obviously $|V(\ma{A})|=q^{N-1}$. Let $\pi_0,\pi_1,\ldots,\pi_{N-1}$ be the maps defined as above. Then our desired perfect hash family can be represented as the following matrix

        \begin{center}

        $\left(
          \begin{array}{ccccc}
            \pi_0(\ma{A})     & \pi_{N-1}(\ma{A}) & \cdots & \cdots & \pi_1(\ma{A}) \\
            \pi_1(\ma{A})     & \pi_0(\ma{A})     & \cdots & \cdots & \pi_2(\ma{A}) \\
            \vdots       & \vdots       & \ddots &        & \vdots   \\
            \vdots       & \vdots       &        & \ddots & \vdots   \\
            \pi_{N-1}(\ma{A}) & \pi_{N-2}(\ma{A}) & \cdots & \cdots & \pi_0(\ma{A}) \\
          \end{array}
        \right)$,

        \end{center}

        \noindent where for every $0\le i\le N-1$, $\pi_i(\ma{A}):=(\pi_i(\alpha))_{\alpha\in V(\ma{A})}$ is a $1\times|V(\ma{A})|$ submatrix. Denote this representation matrix as $M$, then $M$ is an $N\times Nq^{N-1}$ matrix. Let $Y=\cup_{i=0}^{N-1}\pi_i(\ma{A})$ denote the alphabet set. It is not hard to see $|\{\pi_0(\alpha):\alpha\in V(\ma{A})\}|=q^{N-1}$ and $|\{\pi_i(\alpha):\alpha\in V(\ma{A})\}|=q^{N-2}$ for every $1\le i\le N-1$. Then one can verify that $|Y|=q^{N-1}+(N-1)q^{N-2}$. Thus we can conclude that $M$ is the representation matrix of an $(N;Nq^{N-1},q^{N-1}+(N-1)q^{N-2})$-hash family.

        Now it remains to verify that this hash family is indeed an $(N+1)$-perfect hash family. Consider $M$ as the  concatenation of $N$ column patterns denoted as $(C_1|C_2|\cdots|C_N)$ with $|C_1|=|C_2|=\cdots=|C_N|=q^{N-1}$. Take an arbitrary $(N+1)$-subset $S$ of the columns of $M$. We are going to show that there must exist a row of $M$ that separates $S$. If $S\s C_i$ for some $1\le i\le N$, then the $i$-th row of $C_i$, which corresponds to $\pi_0$, can separate $S$, since $\pi_0(\alpha)\neq\pi_0(\beta)$ for arbitrary distinct $\alpha,~\beta\in V(\ma{A})$. Otherwise, let $C_{i_1},\ldots,C_{i_j}$ be the column patterns which have non-empty intersection with $S$, where $j\ge 2$ is a positive integer. For $1\le l\le j$, denote $C_{i_l}\cap S=S_l$. Then $\sum_{l=1}^{j}|S_l|=N+1$ and $|S_l|\le N$ for every $l$. By Lemma \ref{lastlemma}, at most $|S_l|-1$ rows of $C_{i_l}$ can not separate $S_l$. Since $\sum_{l=1}^{j}(|S_l|-1)=N+1-l\le N+1-2=N-1<N$, then there must exist a row of $(C_1|C_2|\cdots|C_N)$ that separates $\cup_{i=1}^{l} S_l=S$.

    \end{proof}

    \begin{remark}
        Our construction has an important property satisfying

        $$\lim_{q\rightarrow\infty}\fr{Nq^{N-1}}{q^{N-1}+(N-1)q^{N-2}}=N$$

        \noindent and hence it is asymptotically optimal since we have $p_{N+1}(N,q)\le Nq$ by Theorem \ref{recusivebd}. Note that a $u$-perfect hash family is $\{w_1,\ldots,w_t\}$-separating for arbitrary $w_i$ such that $w_i\ge1$ and $\sum_{i=1}^t w_i=u$. Theorems \ref{recusivebd} and \ref{const} can be combined to show

        $$\lim_{q\rightarrow\infty}\fr{C(u-1,q,\{w_1,\ldots,w_t\})}{q}=u-1,$$

        \noindent which gives a negative answer to Question \ref{trung}. Furthermore, taking into account the fact that any $(\lfloor(t/2+1)^2\rfloor)$-perfect hash family is also a $t$-IPP code, one can see that our construction also confirms the validity of Conjecture \ref{alonstav}.
    \end{remark}

    \begin{remark}
        It is worth mentioning that Proposition 2 of \cite{BL03perfect} (an unpublished paper) also noticed that $\lim_{q\rightarrow\infty}\fr{p_u(u-1,q)}{q}=u-1$. The author used an optimization method and no explicit construction was given in that paper.
    \end{remark}

    The proof of Lemma \ref{construction} also leads to a conclusion on Hamming graphs, which we think may be of independent interest.

    \begin{corollary}\label{hamming}
        Color the edges of $H(k,q)$ with $k$ colors such that the edge $(\alpha,\beta)$ is colored by color $i$ if $\alpha$ and $\beta$ differ in their $i$-th coordinate. Then $H(k,q)$ contains no cycles with pairwise distinct colors.
     \end{corollary}

\section{Perfect hash families of strength three with three rows}

    Constructions for perfect hash families can induce constructions for corresponding separating hash families. And with the aid of Lemma \ref{technical}, upper bounds for perfect hash families can also induce upper bounds for related separating hash families. Therefore, from this section we will focus on perfect hash families.

   We have mentioned in Section 1 that if $(u-1)\nmid N$, it is very difficult to determine whether the exponent $\lc N/(u-1)\rc$ in Theorem \ref{recusivebd} is tight. In the following two sections we will handle two small cases in such problems, namely, $N=u=3$ and $N=u=4$. When $N=3$ and $u=3$, the corresponding separating hash families only have two alternative types, namely, $\{1,2\}$-separating and 3-perfect hashing. Bazrafshan and Trung  \cite{trung2014} proved that $C(3,q,\{1,2\})\le q^2$ and an $SHF(3;q^2,q,\{1,2\})$ does exist for $q\ge2$. Walker and Colbourn \cite{N=3t=3} conjectured that $p_3(3,q)=o(q^2)$. In this section, we will verify this conjecture by proving $q^{2-o(1)}<p_3(3,q)=o(q^2)$. Furthermore, the upper bound is extended to $p_t(t,q)$ and $C(u,q,\{w_1,\ldots,w_t\})$ with $\sum_{i=1}^t w_i=u$.


  Let us begin with a simple lemma. Note that we will not distinguish between a perfect hash family and its representation matrix. We say a word $x$ of the hash family (resp. a column of the representation matrix) has a unique coordinate $i$ if for any other word (resp. column) $y$, $y\neq x$, it holds that $y(i)\neq x(i)$.

  \begin{lemma}\label{uniquelemma}
     Let $X$ denote the column set (words) of a $PHF(N;n,q,t)$. Then by deleting at most $Nq$ words from $X$, we can get a subset $X^{*}\s X$ such that no word in $X^{*}$ has a unique coordinate in $X^{*}$.
  \end{lemma}

  \begin{proof}
    We use a greedy algorithm to construct $X^{*}$. Delete $x_1$ from $X$ if $x_1$ has a unique coordinate in $X$. Denote $X_1=X-\{x_1\}$. In general, if $x_{i+1}\in X_i$ has a unique coordinate in $X_i$, we delete $x_{i+1}$ from $X_i$ and then denote $X_{i+1}=X_i-\{x_{i+1}\}$. Continue this procedure until we get an $X^{*}$ with no words containing a unique coordinate in it. At most $Nq$ words will be deleted from $X$ since we can delete any symbol $y\in[q]$ at most one time for any coordinate $i\in[N]$.
  \end{proof}

  Since all perfect hash families being considered in the following are of size at least $q^{1+\epsilon}$ for some positive constant $\epsilon$, then the deletion of at most $Nq$ words from $X$ can be neglected. Let $PHF^*(N;n,q,t)$ denote the perfect hash family (obtained from $PHF(N;n,q,t)$) such that no word in it contains a unique coordinate. We use $p_t^*(N,q)$ to denote the corresponding maximal cardinality.

  \begin{lemma}\label{linear}
    In a $PHF^*(t;n,q,t)$, any two words can agree with at most one coordinate.
  \end{lemma}

  \begin{proof}
    Assume the contrary, then the following submatrix is contained in the representation matrix of such $PHF^*(t;n,q,t)$

     $$\left(
          \begin{array}{cccccc}
            \bm{\alpha_1(1)} & \bm{\alpha_2(1)} & *         & *      & *      & * \\
            \bm{\alpha_1(2)} & \bm{\alpha_2(2)} & *         & *      & *      & * \\
            \bm{\alpha_1(3)} & *         & \bm{\alpha_3(3)} & *      & *      & * \\
            \vdots    & *         & *         & \ddots & *      & * \\
            \vdots    & *         & *         & *      & \ddots & * \\
            \bm{\alpha_1(t)} & *         & *         & *      & *      & \bm{\alpha_t(t)} \\
          \end{array}
        \right),$$

     \noindent where in each row, the two bold coordinates are equal. $\alpha_1,\alpha_2$ are two words such that $\alpha_1(i)=\alpha_2(i)$ for $i=1,2$ and since $\alpha_1$ has no unique coordinates, there exist $\alpha_3,\ldots,\alpha_t$ such that $\alpha_j(j)=\alpha_1(j)$ for each $3\le j\le t$. Therefore, no row of the submatrix can separate $\{\alpha_1,\ldots,\alpha_t\}$, violating the $t$-perfect hashing property.
  \end{proof}

  The following two observations are very useful.

  \vspace{10pt}

    \textbf{Observation 1.} On one hand, any $N\times n$ $q$-ary matrix $M$ can be viewed as an $N$-uniform $N$-partite hypergraph $\ma{G}=(V(\ma{G}),E(\ma{G}))$ with equal part size $q$, where the vertex set is defined as $V(\ma{G})=\cup_{i=1}^N V_i$, $V_i=\{(i,j):1\le j\le q\}$ for $1\le i\le N$, and the edge set is defined as $E(\ma{G})=\{\{(i,x(i))\}_{i=1}^N:x=\{x(i)\}_{i=1}^N~is~a~column~of~M\}$.

  \vspace{10pt}

    \textbf{Observation 2.} On the other hand, given an $N$-uniform $N$-partite hypergraph $\ma{G}=(V(\ma{G}),E(\ma{G}))$ with equal part size $q$. We can regard $E(\ma{G})$ as some $N\times|E(\ma{G})|$ $q$-ary matrix $M$. Note that $V(\ma{G})$ can be partitioned into $N$ pairwise disjoint sets with size $q$. We can set $V_i=\{(i,j):1\le j\le q\}$ for $1\le i\le N$, where the first coordinate $i$ corresponds to the $i$-th part $V_i$ and the second coordinate $j$ corresponds to the $j$-th vertex in $V_i$. Then the matrix $M$ is formed by setting its column set as $\{x=\{x(i)\}_{i=1}^N:\{(i,x(i))\}_{i=1}^N\in E(\ma{G})\}$. Such $M$ is said to be the representation matrix of $E(\ma{G})$.

  \vspace{10pt}

   These two observations establish a bridge between $q$-ary matrices and multipartite hypergraphs. Recall the definition of $f_r^*(n,v,e)$ in Section 2.

  \begin{lemma}\label{f3*=p3*}
    $p_3^*(3,q)\le f_3^*(q,6,3)\le p_3(3,q).$
  \end{lemma}

  \begin{proof}

  It is not hard to see that a $PHF^*(3;n,q,3)$ exists if and only if the following configuration is not contained in its representation matrix

  $$\left(
    \begin{array}{ccc}
      a & * & a \\
      b & b & * \\
      * & c & c \\
    \end{array}
  \right),$$\noindent where none of the stars belong to $\{a,b,c\}$.

  We call this configuration a triangle since these three columns have no identical coordinates and every pair of columns have exactly one common coordinate. On one hand, it holds that $f_3^*(q,6,3)\le p_3(3,q)$, since for arbitrary three columns of a hash family, if no row can separate them then for each row there exists some coordinate equal to another one. Therefore, these columns (or corresponding edges) must be spanned by at most six points, which violates the (6,3)-free property. On the other hand, if some three columns of a $PHF^*(3;n,q,3)$ contain at most six points, then either there exists a pair of two columns having two coordinates in common or these three columns form a triangle. Both cases are forbidden in a $PHF^*(3;n,q,3)$. Therefore, it holds that $p_3^*(3,q)\le f_3^*(q,6,3)$ and hence our lemma follows.
  \end{proof}

  \begin{theorem}\label{p_3(3,q)=f_3(3q,6,3)+O(q)}
    $p_3(3,q)=f_3(3q,6,3)+O(q)$ and hence for arbitrary $\epsilon>0$, $q^{2-\epsilon}<p_3(3,q)=o(q^2)$ holds for sufficiently large $q$.
  \end{theorem}

  \begin{proof}
    Apply Lemmas \ref{erdos-kleitman}, \ref{uniquelemma}, \ref{f3*=p3*} and inequality (\ref{1}).
  \end{proof}

  As the second application of Lemma \ref{technical}, the upper bound of $p_3(3,q)$ can be extended to $p_t(t,q)$ and $C(u,q,\{w_1,\ldots,w_t\})$.
  \begin{corollary}\label{cccc}
    $C(u,q,\{w_1,\ldots,w_t\})=o(q^2)$ for any $t\ge 3$ and $\sum_{i=1}^t w_i=u$. In particular, $p_t(t,q)=o(q^2)$ for any $t\ge 3.$
  \end{corollary}

  \begin{proof}
    Apply Lemma \ref{technical} and Theorem \ref{p_3(3,q)=f_3(3q,6,3)+O(q)}.
  \end{proof}

  \begin{remark}
    One can also prove $p_t(t,q)=o(q^2)$ by applying the graph removal lemma \cite{graphremoval}, see \cite{ippe,ippc} for examples of applications of graph removal lemma in such problems. Here our proof applying Lemma \ref{technical} is much simpler. When $1+w\le q$, it was shown in \cite{trung2014} that $C(1+w,q,\{1,w\})\le q^2$. And for any prime power $q$, there exists an $SHF(w+1;q^2,q,\{1,w\})$. Therefore, for $C(u,q,\{w_1,\ldots,w_t\})$ with $t=2$, we can not determine whether $C(w_1+w_2,q,\{w_1,w_2\})=\Omega(q^2)$ or $C(w_1+w_2,q,\{w_1,w_2\})=o(q^2)$. It is an interesting problem to determine the right order of the magnitude of $C(w_1+w_2,q,\{w_1,w_2\})$.
  \end{remark}

  Although we can get the lower bound $q^{2-\epsilon}$ by a direct application of the (6,3)-theorem and Lemma \ref{erdos-kleitman}, we prefer a construction which provides the explicit cardinality. A method introduced in Section 3 of \cite{Furediconst} can be used to construct such $q$-ary codes of length $N$. Our method is similar to that one except some transformations which will be mentioned later.

  Given integers $q\ge N\ge 2$, $M\s\{0,1,\ldots,q-1\}$, we define an $N$-uniform $N$-partite hypergraph $\ma{G}_M$ (whose edge set can be viewed as the representation matrix of our desired code) as follows. The vertex set $V(\ma{G}_M)$ is defined to be

  $$V(\ma{G}_M):=\{(j,y):j\in[N],~y\in\mathbb{Z}_q\}.$$


  \noindent It is easy to see $|V(\ma{G}_M)|=Nq$. For each $j\in[N]$, we use $V_j=\{(j,y):y\in\mathbb{Z}_q\}$ to denote the vertex set of the $j$-th part of $V(\ma{G})$. For integers $0\le y,~m\le q$, the hyperedge of $\ma{G}$ is defined to be the $N$-element set $$A(y,m)=\{(1,y+b_1m),(2,y+b_2m),\ldots,(N,y+b_{N}m)\},$$

  \noindent where $\ma{B}:=\{b_1,\ldots,b_{N}\}\s\{0,1,\ldots,q-1\}$ is an undetermined $N$-element set and the second coordinates $y+b_im$ are taken modulo $q$. We call $\ma{B}$ the tangent set of $A(y,m)$. $A(y,m)$ can also be viewed as a $q$-ary word of length $N$. If $q$ is a prime, one can verify that

  \begin{equation}\label{primeq}
    \begin{aligned}
  |A(y,m)\cap A(y',m')|\le 1
    \end{aligned}
  \end{equation} \noindent holds for $(y,m)\neq (y',m')$ by solving a system of two congruence equations.

  From now on, we fix the size of the alphabet set $q$ to be a prime or the prime nearest to it. For a subset $M\s\{0,1,\ldots,q-1\}$, we set $$E(\ma{G}_M):=\{A(y,m):y\in\mathbb{Z}_q,~m\in M\}$$

  \noindent to be the edge set of our desired hypergraph, where the set $M$ is determined by the subgraphs that needed to be forbidden (these subgraphs can also be viewed as the configurations that needed to be forbidden in the desired code). Obviously $|E(\ma{G}_M)|=q|M|$ and by (\ref{primeq}) we can verify that $\ma{G}_M$ is also linear.

  Now, we are going to choose appropriate $\ma{B}$ and $M$ according to the properties of our desired codes. For example, to construct a 3-perfect hash family with three rows, we first set $N=3$ and then choose $\ma{B}=\{0,1,2\}$, where $b_i=i-1$ for $1\le i\le 3  $. Therefore, to show this specified $\ma{G}_M$ can indeed induce a $PHF(3,n,q,3)$, by Lemma \ref{f3*=p3*} we only need to guarantee that $E(\ma{G})$ is triangle-free, since it is already linear (we have set $q$ to be a prime). We claim that it suffices to choose $M\s\{0,1,\ldots,\lf (q-1)/2\rf\}$ to be a $2$-sum-free set such that the equation $m_1+m_2=2m_3$ has no solution except $m_1=m_2=m_3$.

  \begin{theorem}\label{optimal3}
    There exists a constant $\gamma$ such that $p_3(3,q)>q^2e^{-\gamma\sqrt{\log q}}$.
  \end{theorem}

  \begin{proof}
    It suffices to show $\ma{G}_M$ contains no triangles for arbitrary 2-sum-free set $M\s\{0,1,\ldots,\lf (q-1)/2\rf\}$. If otherwise, assume that $\{A(y_i,m_i)\in\ma{G}_M$:~$1\le i\le 3\}$ forms a triangle. One can verify that the vertices of this triangle must locate on different parts of $V_1,V_2,V_3$. Thus we can assume that

    \begin{equation*}
    \left\{
    \begin{aligned}
      A(y_1,m_1)&\cap A(y_2,m_2)=\{(j_2,a_2)\}\\
      A(y_2,m_2)&\cap A(y_3,m_3)=\{(j_3,a_3)\}\\
      A(y_3,m_3)&\cap A(y_1,m_1)=\{(j_1,a_1)\}\\
    \end{aligned}
    \right.
    \end{equation*}

   \noindent where $\{j_1,j_2,j_3\}=\{1,2,3\}$ and $a_1,a_2,a_3$ are some positive integers. Then the following three equations hold simultaneously

    \begin{equation*}
    \left\{
    \begin{aligned}
      y_1+(j_2-1)m_1&\equiv y_2+(j_2-1)m_2~\pmod{q}\\
      y_2+(j_3-1)m_2&\equiv y_3+(j_3-1)m_3~\pmod{q}\\
      y_3+(j_1-1)m_3&\equiv y_1+(j_1-1)m_1~\pmod{q}.\\
    \end{aligned}
    \right.
    \end{equation*}

  \noindent  Because of the symmetry of a triangle, we can always assume that $j_1<j_2<j_3$. By a simple elimination we can infer $$(j_2-j_1)m_1+(j_3-j_2)m_2\equiv(j_3-j_1)m_3\pmod{q},$$ or simply $$m_1+m_2\equiv2m_3\pmod{q}.$$ This implies $m_1+m_2=2m_3$ since $m_i\le\lf (q-1)/2\rf$ for all $1\le i\le3$, which contradicts the fact that $M$ is 2-sum-free. By Lemma \ref{additive} there exists a 2-sum-free set $M$ with $|M|>qe^{-\gamma\sqrt{\log q}}$ for some constant $\gamma$. Therefore, it follows that $|E(\ma{G}_M)|=q|M|>|M|>q^2e^{-\gamma\sqrt{\log q}}$.
  \end{proof}

  \section{Perfect hash families of strength four with four rows}

   It is much more complicated to construct 4-perfect hash families such that $p_4(4,q)>q^{2-o(1)}$. We will use the notion of rainbow cycles and $R$-sum-free sets defined in Section 2. In fact, we are going to prove the following result:

   \begin{lemma}\label{rainbow1}
        $p_t^*(t,q)\le g_t^*(q)\le p_t(t,q)$.
   \end{lemma}

   \begin{proof}
        First we are going to show that any $PHF^*(t;n,q,t)$ can induce a $t$-uniform $t$-partite linear hypergraph $\ma{G}$ containing no rainbow cycles. Let $M$ denote the representation matrix of the hash family, then $M$ can also be viewed as the representation matrix of $E(\ma{G})$ by Observation 1. $M$ (resp. $E(\ma{G})$) is already linear by Lemma \ref{linear}. It suffices to show $M$ (resp. $E(\ma{G})$) contains no rainbow cycles. Assume otherwise, the columns (resp. hyperedges) of $M$ (resp. $E(\ma{G})$) indexed by $\alpha_1,\ldots,\alpha_k$ form a rainbow $k$-cycle $v_1,\alpha_1,v_2,\alpha_2,\ldots,v_k,\alpha_k,v_1$ with $k\le t$.
        By Observation 1, the $i$-th part of $V(\ma{G})$ can be defined as $V_i=\{(i,j):~j\in[q]\}$, where the first coordinate corresponds to the $i$-th row of $M$ and the second coordinate corresponds to the $j$-th element in $[q]$. Without loss of generality, we can assume that $v_i$ is from the $i$-th part of the vertex set. Then it holds that $\alpha_i(i)=\alpha_{i+1}(i)$ for $1\le i\le k-1$ and $\alpha_k(k)=\alpha_1(k)$. The following submatrix induced by such $k$-cycle is contained in $M$:

        $$\left(
         \begin{array}{ccccccccccc}
           \bm{\alpha_1(1)}   & \bm{\alpha_2(1)} & \alpha_3(1)      & \alpha_4(1)      &        &  \alpha_{k-1}(1)  & \alpha_k(1)          \\
           \alpha_1(2)        & \bm{\alpha_2(2)} & \bm{\alpha_3(2)} & \alpha_4(2)      &        &  \alpha_{k-1}(2)  & \alpha_k(2)          \\
           \alpha_1(3)        &                  & \bm{\alpha_3(3)} & \bm{\alpha_4(3)} &        &  \alpha_{k-1}(3)  & \alpha_k(3)          \\
           \vdots             &                  &                  & \ddots           &        &      \vdots       & \vdots               \\
           \vdots             &                  &                  &                  & \ddots & \bm{\alpha_{k-1}(k-1)}&\bm{\alpha_k(k-1)} \\
           \bm{\alpha_1(k)}   &                  &                  &                  &        & & \bm{\alpha_k(k)}        \\
         \end{array}
       \right),$$

   \noindent where in each row, the two bold coordinates are equal. Note that in this matrix, the columns represent the hyperedges and the coordinates in each column represent the vertices contained in the corresponding hyperedge. It is easy to see none of the first $k$ rows of $M$ can separate $\{\alpha_1,\ldots,\alpha_k\}$. Note that no column of $M$ has unique coordinates, then there exist $\alpha_{k+1},\ldots,\alpha_t$ such that $\alpha_j(j)=\alpha_1(j)$ for $k+1\le j\le t$, which can also be depicted by

   $$\left(
          \begin{array}{cccccc}
            \bm{\alpha_1(k+1)} & \bm{\alpha_{k+1}(k+1)} & *         & *      & *      & * \\
            \bm{\alpha_1(k+2)} & *         & \bm{\alpha_{k+2}(k+2)} & *      & *      & * \\
            \vdots    & *         & *         & \ddots & *      & * \\
            \vdots    & *         & *         & *      & \ddots & * \\
            \bm{\alpha_1(t)} & *         & *         & *      & *      & \bm{\alpha_t(t)} \\
          \end{array}
        \right).$$

   \noindent Therefore, the left $t-k$ rows of $M$ can not separate $\{\alpha_1,\alpha_{k+1},\ldots,\alpha_t\}$. So we can conclude that no row of $M$ can separate $\{\alpha_1,\ldots,\alpha_t\}$, violating the $t$-perfect hashing property.

   It remains to show that any $t$-uniform $t$-partite linear hypergraph (with equal part size $q$) $\ma{G}$ without rainbow cycles can induce a $PHF(t;n,q,t)$ such that $n=|E(\ma{G})|$. We also use $M$ to denote the representation matrix of $E(\ma{G})$. We claim that if there exists a $t\times t$ submatrix $T$ of $M$ such that no row can separate it, then the hypergraph induced by $T$ will contain a rainbow $k$-cycle with $k\le t$.

   We will argue by induction on $t$. When $t=2$, a $2\times 2$ submatrix can always be separated by one of its two rows provided that the two columns of this submatrix are distinct. When $t=3$, if a $3\times 3$ submatrix of a 3-uniform 3-partite linear hypergraph can not be separated by one of its three rows, then this submatrix actually forms a triangle defined in Lemma \ref{f3*=p3*}. One can verify that this triangle can be represented as a rainbow 3-cycle $\{a,E_1,b,E_2,c,E_3\}$ for some edges $E_1,E_2,E_3$.
   Now assume the statement is true for $t-1$. Take a $t\times t$ matrix $T$ with columns indexed by $C=\{\alpha_1,\ldots,\alpha_t\}$ and rows indexed by $R=\{r_1,\ldots,r_t\}$ such that no row can separate $C$. We denote $C_i=C-\{\alpha_i\}$ and $R_i=R-\{r_i\}$ for each $1\le i\le t$. Furthermore, we use $T_{ij}$ to denote the $(t-1)\times(t-1)$ submatrix formed by $R_i$ and $C_j$. Then for any submatrix $T_{ij}$, there must exist a row that separates all columns of $T_{ij}$ since otherwise $T_{ij}$ contains a rainbow $k$-cycle with $k\le t-1$ by the induction hypothesis.

   Without loss of generality, assume $r_1$ separates $C_t$. Note that this row can not separate $C$, so we can assume further that $\alpha_t(1)=\alpha_1(1)$. Then consider $T_{11}$, there exists a row in $R-\{r_1\}$ that separates $C-\{\alpha_1\}$. We can set this row to be $r_2$. Similarly, there exists $2\le j\le t$ such that $\alpha_1(2)=\alpha_j(2)$ since $r_2$ can not separate $C$. Then $j\neq t$ since $\alpha_1$ and $\alpha_t$ have already agreed on one coordinate, say, $\alpha_t(1)=\alpha_1(1)$. Assume that $\alpha_1(2)=\alpha_2(2)$. Now consider $T_{22}$, then there exists a row in $R-\{r_2\}$ that separates $C-\{\alpha_2\}$. Note that this row can not be $r_1$ since $\alpha_1$ and $\alpha_t$ agree on their first coordinate. We can set this row to be $r_3$. For the same reason, there exists $j\in[t],j\neq2$ such that $\alpha_2(3)=\alpha_j(3)$. Then $j\neq 1$ since it already holds $\alpha_1(2)=\alpha_2(2)$. If $j=t$, we are done since $\{\alpha_1,\alpha_2,\alpha_t\}$ forms a rainbow 3-cycle. So we can set $j=3$.

   The above discussion can be depicted by the following matrix:

    $$\left(
     \begin{array}{ccccccccccccccccc}
       \bm{\alpha_1(1)} & \alpha_2(1) & \alpha_3(1) & \alpha_4(1) & \cdots & \cdots & \alpha_{t-1}(1) & \bm{\alpha_t(1)} \\
       \bm{\alpha_1(2)} & \bm{\alpha_2(2)} & \alpha_3(2) & \alpha_4(2) & \cdots & \cdots & \alpha_{t-1}(2) & \alpha_t(2) \\
       \alpha_1(3) & \bm{\alpha_2(3)} & \bm{\alpha_3(3)} & \alpha_4(3) & \cdots & \cdots & \alpha_{t-1}(3) & \alpha_t(3) \\
       \alpha_1(4) & \alpha_2(4) & \bm{\alpha_3(4)} & \bm{\alpha_4(4)} & \cdots & \cdots & \alpha_{t-1}(4) & \alpha_t(4) \\
                   &             &             &   \ddots    &        &        &                 &              \\
                   &             &             &             & \ddots &        &                 &              \\
       \alpha_1(i-1) & \cdots & \bm{\alpha_{i-2}(i-1)}  &\bm{\alpha_{i-1}(i-1)} & & \cdots & \cdots & \alpha_t(i+1) \\
       \alpha_1(i) & \cdots & &  \bm{\alpha_{i-1}(i)} & \bm{\alpha_{i}(i)} &  &  \cdots&  \alpha_t(i+1) \\
       \alpha_1(i+1) & \cdots & &  & \bm{\alpha_{i}(i+1)} & \bm{\alpha_{i+1}(i+1)} & \cdots & \alpha_t(i+1) \\
                   &             &             &   \ddots    &        &        &                 &              \\
                   &             &             &             & \ddots &        &                 &              \\
     \end{array}
   \right),$$

   \noindent where in each row, the two bold coordinates are equal. We continue this procedure for $T_{i,i}$ with $i\ge3$. By our choice, for all $1\le j\le i$, in row $r_j$ it holds that $\alpha_{j-1}(j)=\alpha_j(j)$ ($\alpha_0$ is recognised as $\alpha_t$). Thus no row in $\{r_1,\ldots,r_{i}\}$ can separate $T_{i,i}$. We can always assume that $r_{i+1}\in R-\{r_i\}$ is the row that separates $C-\{\alpha_i\}$. Then there exists a $j\in[t],j\neq i$ such that $\alpha_i(i+1)=\alpha_j(i+1)$ since $r_{i+1}$ can not separate the whole $C$. Obviously, $j\neq i-1$. If $j\in\{1,\ldots,i-2\}$ or $j=t$, then such choice of $j$ will induce a rainbow $(i-j+1)$-cycle formed by $\{\alpha_j,\ldots,\alpha_i\}$

   $$\left(
          \begin{array}{cccccccccc}
                \bm{\alpha_j(j+1)} & \bm{\alpha_{j+1}(j+1)} & * &  & *  &  *    \\
                * & \bm{\alpha_{j+1}(j+2)} & \bm{\alpha_{j+2}(j+2)} &  & * &  *    \\
                \vdots & \vdots & \vdots & \ddots &  &  \vdots   \\
                \vdots & \vdots & \vdots &  & \ddots &    \vdots \\
                \\
              * & * & * &  & \bm{\alpha_{i-1}(i)} & \bm{\alpha_i(i)}  \\
                \bm{\alpha_j(i+1)} & * & * &  & * & \bm{\alpha_{i}(i+1)}   \\
          \end{array}
    \right)$$

   \noindent or a rainbow $(i+1)$-cycle formed by $\{\alpha_1,\ldots,\alpha_i,\alpha_t\}$

       $$\left(
           \begin{array}{cccccccc}
             \bm{\alpha_1(1)} & * & * &  &  &  & \cdots & \bm{\alpha_t(1)} \\
             \bm{\alpha_1(2)} & \bm{\alpha_2(2)} & * &  &  &  & \cdots & * \\
             * & \bm{\alpha_2(3)} & \bm{\alpha_3(3)} &  & &  & \cdots & * \\
             \vdots & \vdots &  & \ddots &  & & \cdots & \vdots \\
             \vdots & \vdots & \vdots &  & \ddots &  & \cdots & \vdots \\
             \\
             * & * & * &  & \bm{\alpha_{i-1}(i)} & \bm{\alpha_i(i)} &  & * \\
             * & * & * &  &  & \bm{\alpha_i(i+1)} & \cdots & \bm{\alpha_t(i+1)} \\
           \end{array}
         \right).$$

    \noindent If neither one of the above cases holds, we can always assume that $j=i+1$ and continue this procedure.

    This procedure will end when it comes to $T_{t-1,t-1}$ with $\alpha_{t-1}(t)=\alpha_t(t)$. Then $\{\alpha_1,\ldots,\alpha_t\}$ will form a rainbow $t$-cycle and our desired contradiction follows.
   \end{proof}

  We can use a similar method as that of the previous section to construct 4-perfect hash family with four rows. However, we can not simply take $\ma{B}=\{0,1,2,3\}$ since such choice will lead to an equation $$2m_1+2m_2-3m_3-m_4=0,$$ whose solution is not easy to determine as suggested by Ruzsa \cite{ruzsa}. In order to show $p_4(4,q)>q^{2-o(1)}$, we should choose $\ma{B}$ more carefully. Recall that we have set $q$ to be a prime.

 \begin{lemma}\label{rainbow2}
     Let $R=\{b_1,\ldots,b_{r}\}\s\{0,\ldots,q-1\}$ be an $r$-element subset with rank $r(R)$. If $M\s\{0,1,\ldots,\lfloor (q-1)/r(R)\rfloor\}$ is an $R$-sum-free set, then the hypergraph defined by $$E(\ma{G}_M)=\{A(y,m):y\in\mathbb{Z}_q,~m\in M\},$$ where $A(y,m)=\{(i,y+b_im):b_i\in R\}$, is an $r$-uniform $r$-partite linear hypergraph containing no rainbow cycles.
  \end{lemma}

  \begin{proof}
    First it is easy to see $\ma{G}_M$ is $r$-uniform and $r$-partite with $V(\ma{G}_M)=\cup_{j=1}^r V_j$, where $V_j=\{(j,y):y\in\mathbb{Z}_q\}$, $1\le j\le N$. To see $\ma{G}_M$ is also linear, one just needs to notice that if $|A(y,m)\cap A(y',m')|\ge2$, then there are $b_1,b_2\in R$, $b_1\neq b_2$ such that

    \begin{equation*}
    \left\{
        \begin{aligned}
            y+b_1m &\equiv y'+b_1m' \pmod{q}\\
            y+b_2m &\equiv y'+b_2m' \pmod{q}.\\
        \end{aligned}
    \right.
    \end{equation*}
    \noindent Then we can infer $(b_1-b_2)(m-m')\equiv 0 \pmod{q}$, which is a contradiction with $q$ prime.

    Now it remains to show that $\ma{G}_M$ indeed contains no rainbow cycles. Assume the contrary, it contains a rainbow $k$-cycle with $k\le r$, denoted by $v_1,A(y_1,m_1),v_2,A(y_2,m_2),\ldots,v_k,A(y_k,m_k),v_1$, where $v_i\in V_{j_i}$ and $j_{i_1}\neq j_{i_2}$ for $i_1\neq i_2$ by the definition of a rainbow cycle. The following $k$ equations hold simultaneously:

    \begin{equation*}
    \left\{
        \begin{aligned}
           y_1+b_{j_2}m_1 &\equiv &y_2+b_{j_2}m_2 \pmod{q}\\
           y_2+b_{j_3}m_2 &\equiv &y_3+b_{j_3}m_3 \pmod{q}\\
           \vdots &              \\
           y_{k-1}+b_{j_k}m_{k-1} &\equiv &y_k+b_{j_k}m_k \pmod{q}\\
           y_k+b_{j_1}m_k &\equiv &y_1+b_{j_1}m_1 \pmod{q}.\\
        \end{aligned}
    \right.
    \end{equation*}
    \noindent By a simple elimination, one can infer
    $$(b_{j_2}-b_{j_1})m_1+(b_{j_3}-b_{j_2})m_2+\cdots+(b_{j_k}-b_{j{k-1}})m_{k-1}+(b_{j_1}-b_{j_k})m_k\equiv 0 \pmod{q},$$
    \noindent or
    $$(b_{j_2}-b_{j_1})m_1+(b_{j_3}-b_{j_2})m_2+\cdots+(b_{j_k}-b_{j{k-1}})m_{k-1}+(b_{j_1}-b_{j_k})m_k=0,$$
    \noindent since $m_i\le\lf(q-1)/r(R)\rf$ for each $1\le i\le k$, which implies $m_1=\cdots=m_k$ taking into account the fact that $M$ is $R$-sum-free. Thus $y_1=\cdots=y_k$, which is a contradiction. Therefore, we can conclude that $\ma{G}_M$ contains no rainbow cycles.
  \end{proof}

  Lemmas \ref{rainbow1} and \ref{rainbow2} suggest that we can use tools from additive number theory to construct good perfect hash families. As discussed before Theorem \ref{optimal3}, we use $\ma{B}$ to denote the set of tangents of $A(y,m)$. To construct $PHF(4;n,q,4)$, we take $\ma{B}=\{0,2,5,\mu+5\}$, where $b_0=0,b_1=2,b_2=5$ and $b_3=\mu+5$ with $\mu=\lc 2^{\sqrt{\log q}}\rc$. Note that $\mu=o(q^{\epsilon})$ for arbitrary small constant $\epsilon>0$. By previous lemmas, our goal is to construct a $\ma{B}$-sum-free subset $M$ of $\mathbb{Z}_q$ with sufficiently large cardinality.
  The desired hyperedge $A(y,m)$ is defined to be

  \begin{equation}\label{Aym}
    A(y,m)=\{(1,y),(2,y+2m),(3,y+5m),(4,y+(\mu+5)m)\}.
  \end{equation}

  \noindent The following lemma (together with Lemma \ref{rainbow2}) shows that if we choose $M$ as the set defined in Lemma \ref{t=4p=4}, then the corresponding $E(\ma{G}_M)$ contains no rainbow cycles.

  \begin{lemma}\label{norainbowcycles1}
    Choose $M$ as the set defined in Lemma \ref{t=4p=4} and let $\ma{B}$ be the 4-element set defined above, then $M$ is $\ma{B}$-sum-free.
  \end{lemma}

  \begin{proof}
    Note that $M$ has no nontrivial solution to all equations in (\ref{3}), then one can verify this lemma directly by definition.
  \end{proof}

  \begin{lemma}\label{norainbowcycles2}
    The hypergraph defined by $$\ma{G}_M=\{A(y,m):y\in\mathbb{Z}_q,~m\in M\},$$ has no rainbow cycles, where $M$ is the set defined in Lemma \ref{t=4p=4} and $A(y,m)$ is defined in (\ref{Aym}).
  \end{lemma}

  \begin{proof}
    Apply Lemmas \ref{rainbow2} and \ref{norainbowcycles1} and note that $r(\ma{B})=\mu+5$.
  \end{proof}

   \begin{theorem}\label{optimal4}
    There exists a constant $\gamma$ such that $p_4(4,q)>q^2e^{-\gamma(\log q)^{3/4}}$.
  \end{theorem}

  \begin{proof}
    Apply Lemmas \ref{t=4p=4}, \ref{rainbow1} and \ref{norainbowcycles2}. Then the theorem follows from
    $$p_4(4,q)\ge g_4^*(q)\ge|E(\ma{G}_M)|=q|M|>q^2e^{-\gamma(\log q)^{3/4}}.$$
  \end{proof}

  \begin{remark}\label{conditions}
    In the above construction of $PHF^*(4,n,q,4)$, we choose the tangent set $\ma{B}$ of the hyperedge $A(y,m)$ to be $\ma{B}=\{0,2,5,\mu+5\}$ with $\mu=\lc 2^{\sqrt{\log q}}\rc$. This choice of $\ma{B}$ has appeared in \cite{ippe}, where the authors used such $\ma{B}$ to construct 2-IPP codes. In this paper we choose the same $\ma{B}$ as they did since in this way we can save the space for proving Lemma \ref{t=4p=4}. Actually, when $|R|=4$ there are many choices of $\ma{B}$ satisfying the following conditions
    \begin{itemize}
      \item [(a)] $M\s\{0,1,\ldots,\lf(q-1)/r(R)\rf\}$ is $R$-sum-free,
      \item [(b)] $|M|>q^{1-o(1)}$,
      \item [(c)] $r(R)=o(q^{\epsilon})$ for arbitrary small $\epsilon>0$.
    \end{itemize}
    \noindent However, for $|R|\ge5$, we do not know whether such $\ma{B}$ exists.
  \end{remark}

\section{Connections to hypergraph Tur{\'a}n problems}

    In this section we will study perfect hash families in view of hypergraph Tur{\'a}n problems.

    \begin{theorem}\label{general}
       For arbitrary positive integers $t,N,q$, it holds that $f_N^*(q,tN-N,t)\le p_t(N,q)$. Furthermore, $\fr{N!}{N^N}f_N(Nq,tN-N,t)\le p_t(N,q)$.
    \end{theorem}
    \begin{proof}
        By Lemma \ref{erdos-kleitman}, it suffices to prove the first statement of the theorem. Recall that if a hypergraph $\ma{G}$ is $N$-uniform $N$-partite with equal part size $q$, then $E(\ma{G})$ can be represented by an $N\times|E(\ma{G})|$ $q$-ary matrix $M$. If $\ma{G}$ is $G(tN-N,t)$-free, then given any collection of $t$ edges $S\s E(\ma{G})$, it is not hard to verify that in its representation matrix there must exist a row that separates $S$, since otherwise $S$ can contain at most $tN-N$ vertices, violating the fact that $\ma{G}$ is $G(tN-N,t)$-free. Therefore, $M$ can be viewed as the representation matrix of the desired perfect hash family.
    \end{proof}

    A direct application of Theorem \ref{general} gives the following result.

    \begin{corollary}\label{cccccccccccccccccccc}
        If $2\nmid N$, then for arbitrary $\epsilon>0$, it holds that $p_3(N,q)>q^{\lc N/2\rc-\epsilon}$.
    \end{corollary}

    \begin{proof}
        This corollary follows from the inequality (\ref{2}), $n^{k-o(1)}<f_r(n,3(r-k)+k+1,3)=o(n^k)$.
        Set $N=2k-1$ and $t=3$, by Theorem \ref{general} one can infer
        $$p_3(N,q)\ge p_3^*(N,q)\ge f_N^*(q,3N-N,3)\ge\fr{N!}{N^N}f_N(Nq,3N-N,3)>\fr{N!}{N^N}(Nq)^{\lc N/2\rc-o(1)}.$$
    \end{proof}

\section{Concluding remarks}

In this paper we mainly study codes and hash families with the separating property. Several open problems and conjectures concerning the upper or lower bounds are solved. Our two essential methods to study these objects can be summarized as follows.

The first method is to discover the structural information hidden in the separating property. As an example, our Johnson-type bound (Lemma \ref{technical}) is used to establish Theorem \ref{recusivebd} and Corollary \ref{cccc}.

The second one is that we establish a bridge between perfect hash families, graph theory and additive number theory. For example, we solve Conjecture \ref{walker} by considering a related hypergraph Tur{\'a}n problem. We also showed that tools from additive number theory can be used to construct good perfect hash families. As a result, Theorems \ref{optimal3}, \ref{optimal4} and Corollary \ref{cccccccccccccccccccc} suggest that there may exist a positive answer to Question \ref{nmid}.

Besides these two new methods, we believe that the construction in Section 4 is of interest since it generalizes many previous ones. Further generalizations of our method are expected.

As a conclusion, we would like to mention several open problems which we think are interesting.

\vspace{10pt}



\textbf{Open Problem 1.}
    If $2\nmid N$, Corollary \ref{cccccccccccccccccccc} shows that $p_3(N,q)>q^{\lc N/2\rc-o(1)}$. Determine whether $p_3(N,q)=o(q^{\lc N/2\rc})$ or $p_3(N,q)=\Theta(q^{\lc N/2\rc})$.

\vspace{10pt}

\textbf{Open Problem 2.}
   For $r$-uniform $r$-partite linear hypergraph without rainbow cycles, we have proved that $g_r^*(q)=o(q^2)$ and $g_i^*(q)>q^{2-o(1)}$ for $i=3,4$. Then does it hold that $g_r^*(q)>q^{2-o(1)}$ for all $r\ge3$?
\vspace{10pt}

\textbf{Open Problem 3.}
    For arbitrary $r\ge3$, does there exist an $r$-element set $R$ and $M\s[q]$ such that the conditions in Remark \ref{conditions} are satisfied? Note that the question is true when $r=3,4$.
\vspace{10pt}

\textbf{Open Problem 4.}
    It has been shown in Theorem \ref{general} that $p_t(N,q)\ge f_N^*(q,tN-N,t)$. Then does there exist an upper bound for $p_t(N,q)$ only using $f_N^*(q,v,t)$?

\bibliographystyle{plain}
\bibliography{shf}

\end{CJK*}
\end{document}